\documentclass[11pt]{article}
\usepackage{axient_ffra_compat}
\usepackage{algorithm,algpseudocode,placeins,url,xspace,seqsplit}
\hypersetup{colorlinks=true,linkcolor=black,citecolor=black,urlcolor=blue!55!black,
  pdftitle={Axient: Canonical Protocol-Graph Composition for Leveraged Event Markets: Single State Authority, Atomic Composition, Durable Sagas, and Exactly-Once Recovery},
  pdfauthor={Maksym Nechepurenko},
  pdfsubject={Protocol composition and financial-state authority for leveraged event markets},
  pdfkeywords={event markets, leveraged prediction markets, protocol composition, smart contracts, on-chain credit, settlement evidence, liquidation, loss waterfall, exactly-once recovery}}
\selectfont
\newcolumntype{P}[1]{>{\raggedright\arraybackslash}p{#1}}
\crefname{definition}{definition}{definitions}
\crefname{assumption}{assumption}{assumptions}
\crefname{proposition}{proposition}{propositions}
\crefname{theorem}{theorem}{theorems}
\crefname{corollary}{corollary}{corollaries}
\crefname{remark}{remark}{remarks}
\definecolor{axblue}{HTML}{4263EB}
\definecolor{axteal}{HTML}{14B8A6}
\definecolor{axnavy}{HTML}{0B172A}
\definecolor{axgray}{HTML}{EFF3F8}

\newcommand{\Hh}{\mathsf{H}}
\newcommand{\interimnote}{\begin{center}\fcolorbox{axteal!70!black}{axteal!6}{\parbox{0.93\textwidth}{\centering\textbf{Registered local evidence result.} One accepted cohort contains twelve uniquely correlated FIA records and 84 retained \texttt{ACTUAL\_EVIDENCE} layer positions. Two clean materializations matched at the registered byte/file-index boundary; a deterministic strict archive was accepted by a separate verifier; and a selected ordered replay is bound to that archive.}}\end{center}}

\title{\textbf{Axient: Canonical Protocol-Graph Composition for Leveraged Event Markets: Single State Authority, Atomic Composition, Durable Sagas, and Exactly-Once Recovery}}
\author{Maksym Nechepurenko\\\small Research Department of Devnull FZCO, Dubai, UAE\\\small \href{mailto:maksym@devnull.ae}{maksym@devnull.ae} \;|\; \href{https://orcid.org/0000-0002-9515-8841}{ORCID: 0000-0002-9515-8841}}
\date{30 September 2026}
\ffraSetup{series={Axient Research \textperiodcentered\ Paper 6},author={Maksym Nechepurenko},affiliation={Research Department of Devnull FZCO, Dubai, UAE},email={maksym@devnull.ae},orcid={0000-0002-9515-8841},version={r0.6.3},date={30 September 2026}}

\begin{document}
\maketitle

\begin{abstract}
A modular leveraged event-market protocol can contain individually correct contracts for risk approval, positions, debt, settlement evidence, credit pools, junior backstops, reserves, liquidation, and governance while still lacking one authoritative financial execution path. This paper develops a canonical protocol graph in which every financial domain has one storage authority and composition components coordinate transitions without recreating balances, debt, positions, evidence receipts, queues, reserves, pause state, or terminal scenario state. The architecture formalizes atomic position origination, a settlement-confirmed debt saga, trader residual release, an ordered reserve--backstop--Senior loss waterfall, loss-participating withdrawal queues, capability-conditioned execution, protocol-wide pause and timelocked recovery, deterministic deployment manifests, canonical state projection, and exactly-once reconstruction after a chain-confirmed checkpoint. We prove composition-level authority uniqueness, atomic closure, loss conservation and priority, receipt-idempotent settlement, legal no-go safety, capability-transition safety, and no-financial-reexecution recovery under explicit assumptions. The implementation is evaluated through twelve parent Financial Interaction Assertions (FIA-001--FIA-012), including mandatory FIA-011A/B capability subcases. The registered local cohort contains twelve uniquely correlated scenario records and 84 retained \texttt{ACTUAL\_EVIDENCE} layer positions, two clean materializations with one byte-identical file index, a deterministic strict archive accepted by a separate verifier, and a selected ordered replay bound to that archive. The result connects the formal protocol graph to a reproducible evidence chain while preserving explicit boundaries between the observed engineering result and future external-operation hypotheses.
\end{abstract}


\section{Introduction}
Leveraged event-market positions combine a time-bounded trading problem with a credit problem. A trader supplies collateral, a credit provider finances additional outcome exposure, an execution venue determines what can be bought or sold, and an oracle or verification gateway later fixes the event payout. Earlier Axient work separates leverage maturity from claim maturity and develops debt-free finality: leverage is extinguished before the claim becomes non-tradable, leaving a fully funded residual claim for later resolution \citep{nechepurenko2026debtfree}. The companion credit paper introduces Senior Credit LPs, Liquidation Backstop Providers (LBPs), reserves, queues, and a deterministic loss waterfall \citep{nechepurenko2026credit}. The empirical programme then distinguishes formal certificate bounds from statistical validation and defines a controlled internal-alpha methodology \citep{nechepurenko2026empirical}.

Those results do not, by themselves, answer an implementation question: when the system is decomposed into many contracts and services, which object is the financial source of truth? A test-only bridge can reproduce a desired terminal state while maintaining its own queue, pause flag, evidence registry, or scenario state. A reference waterfall can be correct algebraically yet become a second capital ledger if it is connected alongside real vault balances. An application service can expose plausible projections without proving that the intended canonical selectors were called. The result may be internally consistent but architecturally ambiguous.

This paper studies the composition problem. Its central requirement is:
\begin{quote}
\emph{Every financial domain has exactly one canonical storage authority. Composition code may coordinate canonical modules, but it may not own a second copy of balances, debt, positions, evidence receipts, queues, reserves, pause state, or terminal financial state.}
\end{quote}

The requirement is stronger than modularity. Modularity says that responsibilities are separated; canonical composition says that their states and transitions jointly form one financial machine. It also differs from a generic orchestrator. The orchestrator may schedule, correlate, retry, and project, but it cannot settle debt by rewriting its own database. The authority boundary is on-chain and manifest-bound.

\paragraph{Contributions.}
\begin{enumerate}[label=(C\arabic*)]
\item A composition-level storage-authority rule that applies to every nonempty subset of financial domains, not only to each component in isolation.
\item Storage-free composition controllers for position origination, ordered loss allocation, and protocol-wide emergency pause.
\item Capability-conditioned execution authority, including stale-capability rejection before admission and fail-closed book-epoch transition handling while debt remains positive.
\item An authority context binding each financial mutation to a legal-entity role, custody class, and permitted-use class.
\item A distinction between atomic EVM transitions and durable, idempotent service sagas.
\item A canonical settlement lifecycle in which typed evidence is the idempotency boundary and debt reduction precedes residual or collateral release.
\item A single executable loss waterfall over reserve, LBP, and Senior modules without a duplicate capital ledger.
\item A deterministic deployment manifest that binds contract addresses, code identities, ABIs, storage layouts, roles, selectors, pause targets, adapter identity, and application schemas.
\item A canonical state projection derived from ABI views and ordered canonical events rather than fixture values.
\item An effect-level recovery criterion based on a chain-confirmed checkpoint and replay without a duplicate committed financial effect.
\item A twelve-scenario FIA evaluation surface and a strict separation between the completed cohort-bound evidence result and successor external-operation hypotheses.
\end{enumerate}

The architecture overview, ordered waterfall, and emergency-control topology are shown in \cref{fig:graph,fig:waterfall,fig:pause}.

\section{System model and evidence scope}
\subsection{Financial domains}
Let $\mathcal D$ be the finite set of financial domains and $\mathcal M$ the set of deployed modules. The relevant domains are
\[
\mathcal D=\{\textsf{risk},\textsf{position},\textsf{debt},\textsf{evidence},\textsf{senior},\textsf{lbp},\textsf{reserve},\textsf{liquidation},\textsf{governance}\}.
\]
The global financial state is the product
\[
X=\prod_{d\in\mathcal D}X_d,
\]
where $X_d$ is the state space of domain $d$.

\begin{definition}[Canonical authority map]
A canonical authority map is a function
\[
A:\mathcal D\rightarrow\mathcal M
\]
that assigns exactly one storage-owning module to each financial domain.
\end{definition}

\begin{definition}[Cross-cutting pause state]
Let $\mathcal T$ be the nonempty manifest-frozen set of pausable modules. Its authoritative pause state is the indexed vector $P_{\mathcal T}=(p_m)_{m\in\mathcal T}$, where each component is owned by its target module. The statement that all targets are paused is distinct from the statement that at least one is paused. An action is governed by its registered target and selector policy, not by an unqualified aggregate flag. Missing or stale required state does not establish permission or global safety. Pause is a cross-cutting control over financial domains, not another financial domain with one storage owner. The coordinator may fan out commands but owns no component of $P_{\mathcal T}$.
\end{definition}

\begin{definition}[Authority context]
An authority context is a tuple
\[
\lambda=(e,c,p),
\]
where $e$ is the legally responsible entity class, $c$ the custody class governing the asset or claim, and $p$ the permitted-use class. A module address is technically reachable but not financially authorized unless the manifest binds the relevant domain to an admissible authority context $\lambda$.
\end{definition}

\begin{definition}[Venue capability state]
At time $t$, the venue capability state is
\[
\kappa_t=(a_t,b_t,\rho_t),
\]
where $a_t$ is the versioned capability vector, $b_t$ the book epoch, and $\rho_t$ the pause/reactivation mode. Capability continuity for a command means that every capability used by its certificate remains valid over the registered execution interval, or that a registered risk-reducing fallback is triggered when continuity fails.
\end{definition}

The registered local composition uses the financial-domain mapping in \cref{tab:authority}. Pause is evaluated separately through the manifest-frozen indexed vector in the cross-cutting pause predicate.

\begin{table}[t]
\centering\small
\caption{Canonical storage-authority map.}
\label{tab:authority}
\begin{tabularx}{\textwidth}{@{}P{0.31\textwidth}Y@{}}
\toprule
Domain & Sole canonical authority \\
\midrule
Risk approvals and $2\times$ limits & \texttt{RiskManager} \\
Position lifecycle and locked collateral & \texttt{PositionManagerCore} \\
Principal, accrued interest, matched/settled capacity & \texttt{DebtManagerCore} \\
Evidence, replay, conflict, recovery & \texttt{SettlementEvidenceRegistry} \\
Senior assets, shares, losses, withdrawal queue & \texttt{CreditPool} \\
LBP assets, shares, encumbrance, impairment, queue & \texttt{LBPVault} \\
Segregated reserve custody & \texttt{ReserveVault} \\
Liquidation request and execution evidence & \texttt{LiquidationManager} \\
Governance operations & \texttt{GovernanceTimelock} \\
\bottomrule
\end{tabularx}
\end{table}

The entity, custody, and permitted-use classes associated with the map are summarized in \cref{tab:authority-context}.

\begin{table}[t]
\centering\small
\caption{Authority context bound to canonical domains. Exact entity names are deployment-specific; the classes are part of the manifest identity.}
\label{tab:authority-context}
\begin{tabularx}{\textwidth}{@{}P{0.23\textwidth}P{0.22\textwidth}P{0.23\textwidth}Y@{}}
\toprule
Domain class & Entity role & Custody class & Permitted-use condition \\
\midrule
Position/debt & protocol or licensed/intermediated operator & collateral-locked or contract-controlled & leverage origination and debt service permitted \\
Venue execution & venue-authorized trader, delegate, or clearing agent & venue account or approved contract account & order entry, cancellation, and settlement permitted \\
Claim custody & claim custodian or contract-controlled outcome account & segregated claim custody & financed claim may be held, redeemed, and subject to lien/control \\
Credit/LBP/reserve & credit pool, backstop vehicle, or reserve owner & segregated protocol capital & loss absorption and withdrawal policy permitted \\
Liquidation & authorized liquidator/default manager & execution-only authority & registered risk-reducing transition permitted \\
\bottomrule
\end{tabularx}
\end{table}

A displayed balance or economically available unit is not automatically eligible capital. If the entity, custody, or permitted-use gate fails, the command receives a typed \texttt{LEGAL\_NO\_GO} result and produces zero financial mutation.

\subsection{Controllers and services}
A controller coordinates calls among canonical modules. A service performs durable orchestration, evidence ingestion, indexing, and presentation. Neither is permitted to become financial authority.

\begin{definition}[Storage-free composition controller]
A controller $C$ is storage-free with respect to financial state if its mutable storage contains no balance, share supply, debt, position, evidence receipt, withdrawal queue, reserve account, pause flag, or terminal financial state. Immutable module addresses and non-authoritative correlation events are allowed.
\end{definition}

\begin{definition}[Durable financial saga]
A durable financial saga is a sequence of manifest-authorized canonical commands whose progress is recorded in an append-only inbox/outbox/journal. Retries read canonical state and may repeat idempotent commands, but no service-owned balance or debt state determines financial completion.
\end{definition}

\subsection{Evidence boundary}
The implementation programme is intentionally local-only: Anvil chain ID 31337, loopback services, emulated execution, test assets, and a hard $2\times$ leverage ceiling. No production network, public testnet, remote RPC, real assets, external venue writes, public liquidity, live credit, DNS, or Vercel production promotion belongs to the evaluated scope. The evidence therefore concerns canonical composition under controlled local conditions, not production readiness. The four principal composition boundaries are summarized in \cref{fig:graph}.

\begin{figure}[t]
\centering
\resizebox{0.98\textwidth}{!}{%
\begin{tikzpicture}[
  mod/.style={draw=axblue!70,fill=axgray,rounded corners=2pt,minimum width=2.7cm,minimum height=0.68cm,align=center,font=\small},
  ctl/.style={draw=axteal!80!black,fill=axteal!7,rounded corners=2pt,minimum width=3.25cm,minimum height=0.68cm,align=center,font=\small},
  svc/.style={draw=axnavy!70,fill=axnavy!6,rounded corners=2pt,minimum width=3.25cm,minimum height=0.68cm,align=center,font=\small},
  arr/.style={-{Latex[length=2.4mm]},thick,draw=axnavy!80}]
\node[ctl] (orig) at (0,0) {Position Origination\\Controller};
\node[mod] (risk) at (4,0) {RiskManager};
\node[mod] (pos1) at (7.6,0) {PositionManagerCore};
\node[mod] (debt1) at (11.2,0) {DebtManagerCore};
\draw[arr] (orig)--(risk); \draw[arr] (risk)--(pos1); \draw[arr] (pos1)--(debt1);
\node[svc] (saga) at (0,-1.35) {Settlement Saga\\+ Journal / Projection};
\node[mod] (evid) at (4,-1.35) {SettlementEvidence\\Registry};
\node[mod] (debt2) at (7.6,-1.35) {DebtManagerCore};
\node[mod] (pos2) at (11.2,-1.35) {PositionManagerCore};
\draw[arr] (saga)--(evid); \draw[arr] (evid)--(debt2); \draw[arr] (debt2)--(pos2);
\node[ctl] (wf) at (0,-2.70) {Loss Waterfall\\Controller};
\node[mod] (res) at (4,-2.70) {ReserveVault};
\node[mod] (lbp) at (7.1,-2.70) {LBPVault};
\node[mod] (credit) at (10.2,-2.70) {CreditPool};
\node[mod] (liq) at (13.5,-2.70) {LiquidationManager};
\draw[arr] (wf)--(res); \draw[arr] (res)--(lbp); \draw[arr] (lbp)--(credit); \draw[arr] (credit)--(liq);
\node[ctl] (pause) at (0,-4.05) {Protocol Pause\\Coordinator};
\node[mod,minimum width=4.4cm] (targets) at (4.6,-4.05) {Manifest-frozen\\8-module target set};
\node[mod] (gov) at (8.9,-4.05) {GovernanceTimelock};
\draw[arr] (pause)--(targets); \draw[arr] (targets)--(gov);
\end{tikzpicture}}
\caption{Canonical protocol graph. Each row shows one composition boundary; financial state remains in the domain modules and the controllers do not become a second ledger.}
\label{fig:graph}
\end{figure}
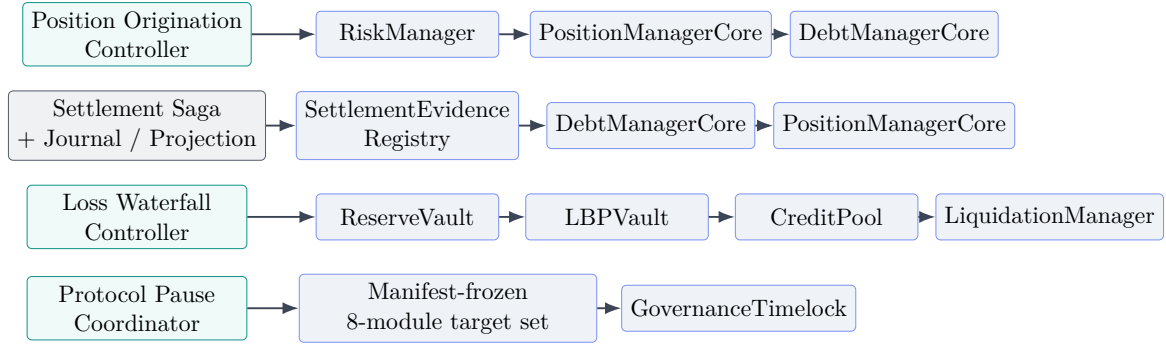

\section{Canonical transition semantics and command admissibility}
\subsection{Global transition system}
Let $x\in X$ be canonical financial state, $m$ a sealed deployment manifest, $r$ a role, and $u$ a typed command. A command transition is a partial function
\[
T_u:X\times M\times R\rightharpoonup X\times O,
\]
where $O$ is the set of receipts and canonical events. Capability, legal-authority, and target-owned pause context are fixed evaluation parameters, suppressed in this notation; they are not inferred from service projections. The function is undefined when any required authorization or precondition fails.

The command is \emph{financially admissible} only if six predicates hold:
\begin{align}
G_M(u,m)&=1 && \text{manifest address and selector authorization},\\
G_R(u,r)&=1 && \text{role and governance authorization},\\
G_P(u,x)&=1 && \text{protocol preconditions in canonical state},\\
G_K(u,x,\kappa)&=1 && \text{venue-capability freshness and continuity},\\
G_L(u,x,\lambda)&=1 && \text{entity, custody, and permitted-use eligibility},\\
G_B(u,x)&=1 && \text{economic and environment boundaries}.
\end{align}
The execution predicate is
\[
G(u,x,m,r,\kappa,\lambda)=G_MG_RG_PG_KG_LG_B.
\]
A failed predicate produces either an EVM revert or a typed service rejection whose evidence explicitly proves zero financial mutation. Economic availability cannot compensate for $G_L=0$, and a previously valid certificate cannot compensate for a capability or book-epoch transition that makes $G_K=0$.

\begin{definition}[Canonical command]
A canonical command is a typed command whose target address and selector are manifest-authorized, whose role is permitted, whose venue-capability and legal-authority contexts are valid for the required interval, and whose financial effect is determined exclusively by the canonical storage authorities.
\end{definition}

\begin{proposition}[No off-manifest equivalence]
Two commands with the same user-facing label but different target addresses, selectors, role policy, or deployment-manifest identities are different canonical commands. Evidence for one does not prove the other.
\end{proposition}

This distinction is important for upgrades and testing. A mock, a legacy bridge, and the canonical position manager may all expose a method called ``open position,'' but their state effects and authority differ. The manifest prevents interface similarity from being treated as execution equivalence.

\begin{proposition}[Legal no-go safety]
If $G_L=0$ and the command dispatcher and canonical modules enforce the manifest-bound legal context, then displayed collateral or credit capacity cannot originate a position, increase debt, or enlarge risk. The rejected command contributes zero financial mutation. Equality of the complete pre- and post-state additionally presumes an isolated execution without unrelated concurrent transitions.
\end{proposition}

\begin{proposition}[Capability-transition safety]
Suppose admission requires fresh capability state $\kappa_{t_0}$ and the active position remains bound to book epoch $b_{t_0}$. If capability is stale before admission, the command is rejected with zero mutation. If $b_t\neq b_{t_0}$ while debt $D_t>0$, risk-increasing authority is revoked and the system must enter a registered fail-closed or risk-reducing path before any further exposure increase.
\end{proposition}

\subsection{Financial state and projection state}
Let $x$ denote canonical state and $y$ denote service/application projection state. The system evolves as
\[
x_{t+1}=T_u(x_t),\qquad y_{t+1}=P(J(x_{0:t+1})),
\]
where $J$ is the append-only journal of canonical observations and $P$ is a deterministic materializer. There is deliberately no transition of the form $x_{t+1}=F(y_t)$ for ordinary projection state. A service can propose a command based on $y_t$, but the contract preconditions and canonical state decide whether it commits.

\begin{proposition}[Projection rebuildability]
If the journal retains the manifest identity, command correlation, transaction and event identifiers, and all inputs required by $P$, then deleting $y$ does not destroy financial state and $y$ can be rebuilt without a financial transaction.
\end{proposition}

\subsection{Module interface map}
\Cref{tab:interfaces} records the intended local composition surface. The public evidence companion records the FIA taxonomy, canonical result identities, and publication-safe registries used by the evaluated path; this article supplies the interface detail needed to establish the composition properties.

\begin{longtable}{@{}P{0.23\textwidth}P{0.24\textwidth}P{0.23\textwidth}P{0.22\textwidth}@{}}
\caption{Canonical module interface and authority map.}\label{tab:interfaces}\\
\toprule
Module & Principal mutations & Canonical events / observations & Composition role \\
\midrule
\endfirsthead
\toprule
Module & Principal mutations & Canonical events / observations & Composition role \\
\midrule
\endhead
RiskManager & approve, consume, release risk capacity & approval, consumption, rejection reason, limits & sole admission authority \\
Position Manager Core & open, transition, release collateral & position opened, state changed, collateral released & position/collateral authority \\
Debt Manager Core & create debt, record match, apply settlement, residual accounting & debt created, match recorded, settlement applied & debt authority \\
Settlement Evidence Registry & configure adapter, submit, classify replay/conflict, record recovery & admitted, replay, exception, recovery & evidence authority \\
CreditPool & Senior deposit, queue/service withdrawal, apply approved loss & shares, queue, service, loss events & Senior capital authority \\
LBPVault & deposit, encumber, impair, queue/service & shares, encumbrance, impairment, queue & junior capital authority \\
ReserveVault & fund and allocate approved loss & funding and reserve allocation & reserve custody authority \\
Liquidation Manager & request, record execution, record allocation hash & request, execution, and allocation evidence & liquidation evidence authority \\
Governance Timelock & queue and execute authorized operations & operation queued/executed & governed recovery authority \\
Position Origination Controller & consume approval, open position, create debt atomically & correlation metadata only & zero-storage coordinator \\
Loss Waterfall Controller & call canonical loss selectors atomically & allocation amounts and digest & zero-storage coordinator \\
Protocol Pause Coordinator & pause frozen target set atomically & fan-out correlation only & zero-state coordinator \\
\bottomrule
\end{longtable}

\subsection{Cross-module invariants}
The graph must preserve the following cross-module invariants:
\begin{enumerate}[label=(I\arabic*)]
\item every open financed position has one bound debt record and a consumed approval;
\item collateral cannot be released while canonical debt remains positive;
\item debt cannot decrease without admitted settlement evidence;
\item one evidence receipt cannot create two financial deltas;
\item trader residual is released only after debt and admitted fees;
\item reserve, LBP, and Senior loss application is ordered and conservative;
\item queued shares remain in the relevant share supply until service and burn;
\item a pause coordinator cannot report a state different from its targets;
\item recovery cannot create a duplicate committed financial effect in the same operation domain;
\item no service projection or controller becomes a financial ledger;
\item capability freshness and book-epoch continuity condition every execution authority;
\item legal entity, custody class, and permitted use are manifest-bound and can force a zero-mutation \texttt{LEGAL\_NO\_GO}.
\end{enumerate}
These invariants form the implementation-level specialization of the accounting properties in Papers I and II rather than a replacement for them. Invariant (I8) is the cross-cutting counterpart of financial-domain authority uniqueness: the coordinator cannot substitute for, contradict, or aggregate away the target-owned pause vector $P_{\mathcal T}$.

\section{Single authority and compositional closure}
\begin{definition}[Composition-level no-secondary-ledger property]
For every nonempty subset $S\subseteq\mathcal D$ and every finite composition $\Gamma$ of controllers, adapters, services, and test utilities, the authoritative state on $S$ is the projection of the states stored by $\{A(d):d\in S\}$. No mutable state owned by $\Gamma$ may determine, replace, or override an authoritative read or write on $S$ independently of the canonical authorities; merely consulting and then overriding them is not sufficient.
\end{definition}

\begin{assumption}[Manifest-complete call surface]
Every mutating command is directed to an address and selector present in a sealed deployment manifest. The manifest identifies all canonical authorities, authority contexts, venue-capability requirements, and allowed storage-free controllers.
\end{assumption}

\begin{theorem}[Financial-domain subset compositional authority uniqueness]
Suppose the authority map $A$ is total and single-valued, every financial transition writes only to $A(d)$ for each affected domain $d$, and every component composition satisfies the composition-level no-secondary-ledger property. Then for every nonempty $S\subseteq\mathcal D$, the composed system has at most one authoritative representation of the joint financial state $X_S=\prod_{d\in S}X_d$.
\end{theorem}
\begin{proof}
Assume a nonempty subset $S$ has two authoritative joint representations. One is the canonical product of the states stored by $\{A(d):d\in S\}$. The other must be determined by mutable state in some composition $\Gamma$ or by a different storage-owning module. Because the alternative can determine or override an authoritative read or write on $S$ independently of the complete canonical product, it violates the composition-level no-secondary-ledger property. The argument covers both a second per-domain ledger and a synthetic subset ledger that combines several domains. Correlation metadata, append-only observations, and rebuildable projections do not contradict the theorem because they are explicitly non-authoritative.
\end{proof}

\begin{proposition}[Pause-vector authority and coordinator non-authority]
Let $\mathcal T$ be the nonempty manifest-frozen pause target set. If each $p_m$ is stored only by target module $m$, the coordinator owns no pause state, and fan-out reverts when any target call fails, then the authoritative pause state is exactly the indexed vector $P_{\mathcal T}=(p_m)_{m\in\mathcal T}$. A coordinator event or projection cannot establish a pause state different from its targets.
\end{proposition}
\begin{proof}
Each component is owned by its target module and the coordinator stores none of them. Atomic fan-out prevents a committed strict subset in the successful path. Therefore invariant (I8) reduces pause verification to the target-owned indexed vector and ordered target events.
\end{proof}

\begin{corollary}[Projection non-authority]
A service, API, browser, or audit projection can be stale or unavailable without changing canonical financial state. Conversely, changing a projection alone cannot create debt repayment, collateral release, share issuance, impairment, or a target-owned pause transition.
\end{corollary}

The result provides a testable architecture criterion. A bridge that stores its own scenario state, queue shares, pause flag, or evidence manager violates the criterion even if its terminal hashes equal a fixture. Such a bridge may remain useful as a harness, but it cannot be promoted to the canonical graph.

\section{Atomic origination}
Position origination consumes risk capacity, locks trader collateral, and creates debt. A partial completion is unsafe: an approval could be consumed without a position, a position could exist without debt, or debt could exist without collateral.

Let $a$ be an approval, $p$ a position, and $d$ a debt record. The intended transition is
\[
T_{\mathrm{orig}}:(a,p,d)\mapsto(a',p',d'),
\]
where $a'$ is consumed, $p'$ is open with locked collateral, and $d'$ contains the principal and authorized terms.

\begin{algorithm}[t]
\caption{Atomic canonical origination}
\label{alg:origination}
\begin{algorithmic}[1]
\Require manifest-authorized modules; valid approval $a$; position request $q$
\State verify $q$ references the same pool, market, account, leverage tier, and collateral as $a$
\State call \texttt{RiskManager.consumeApproval}$(a)$
\State call \texttt{PositionManagerCore.openPosition}$(q)$
\State call \texttt{DebtManagerCore.createDebt}$(q)$
\State emit non-authoritative correlation metadata
\Ensure either all canonical calls commit or the EVM reverts the complete transaction
\end{algorithmic}
\end{algorithm}

\begin{theorem}[Atomic origination closure]
If the three canonical calls in \cref{alg:origination} execute in one EVM transaction, each successful call realizes its specified effect, and every failure is propagated so that the complete transaction reverts, then its committed post-state cannot contain exactly one or two of the three effects: approval consumption, position opening, and debt creation.
\end{theorem}
\begin{proof}
EVM transaction semantics apply state updates atomically \citep{wood2014yellowpaper}. If any call reverts and the revert is not swallowed, all prior state changes in the transaction are reverted. Therefore the only reachable outcomes are the complete pre-state or the complete committed post-state.
\end{proof}

The result concerns the committed transaction boundary, not an arbitrary intermediate callback or an external venue action. Caught exceptions and unchecked low-level failure returns do not satisfy its assumptions \citep{solidityerrors2026}. The controller owns immutable addresses only. It does not cache approval state, position state, or debt. The manifest therefore becomes part of the correctness boundary: calling a different address or selector is a different system, not the same origination under another configuration.

\section{Settlement evidence and durable debt finalization}
\subsection{Why settlement is not one transaction}
External execution and settlement can arrive asynchronously. An emulator or venue adapter reports match and settlement evidence; the evidence registry classifies a receipt; the debt manager applies admitted value; the position manager transitions the lifecycle; and residual/collateral release follows only after canonical getters prove finality. Compressing all steps into one transaction is neither always possible nor desirable. The appropriate abstraction is a durable saga \citep{gray1993transaction,kleppmann2017ddia}.

For a unique evidence identifier $e$, let
\[
R(e)\in\{\textsf{new},\textsf{admitted},\textsf{replay},\textsf{conflict},\textsf{recovered}\}
\]
be the registry state. Let $\Delta D(e)$ be the debt reduction applied from $e$.

\begin{assumption}[Receipt and application uniqueness]
A settlement observation is bound to a stable logical settlement domain and payload. Reusing the domain cannot admit conflicting content or authorize another application by changing only a transport identifier. Evidence admission and financial application are distinct facts. In the debt authority, checking the applied status, recording application, and mutating debt form one atomic guarded transition for that domain. The result below is conditional on this application guard; admission uniqueness alone is insufficient.
\end{assumption}

The durable settlement sequence is formalized in \cref{alg:settlement}. The algorithm distinguishes the existing registry status from the debt manager's existing application status; it introduces no new module or receipt type.

\begin{algorithm}[t]
\caption{Durable canonical settlement saga: admission is not application}
\label{alg:settlement}
\begin{algorithmic}[1]
\Require command $c$, evidence $e$, sealed manifest $M$; receipt/application uniqueness
\State append command and its stable idempotency key to the durable inbox
\State submit or reconcile $e$ at \texttt{SettlementEvidenceRegistry}
\If{receipt is conflicting, inadmissible, or its canonical status is uncertain}
  \State retain exception/reconciliation evidence
  \State \Return without a new financial mutation
\EndIf
\State read the canonical application status for the same settlement domain and payload
\If{the settlement effect is already committed}
  \State do not reapply it; read canonical debt and position
\ElsIf{receipt has a valid admission and application is reliably not yet committed}
  \State call the atomically guarded \texttt{DebtManagerCore.applyAdmittedSettlement}
\Else
  \State \Return for reconciliation; do not invent a new key or assume application
\EndIf
\State confirm the application outcome, or return for reconciliation
\State reconcile any still-required permitted \texttt{PositionManagerCore} transition
\State persist the confirmed application checkpoint; rebuild missing projections
\State evaluate residual/collateral release separately under its own canonical eligibility guard
\end{algorithmic}
\end{algorithm}

\begin{theorem}[Receipt-idempotent debt application]
Under receipt and application uniqueness, the same admitted settlement can produce at most one committed debt application. A retry after its successful application produces zero additional debt reduction. An admitted but unapplied receipt remains eligible for its first guarded application; registry replay alone does not establish debt repayment.
\end{theorem}
\begin{proof}
The first successful application atomically records the applied state with the debt mutation. Every concurrent or later call in the same settlement domain therefore either supplies that first application or encounters the already-applied guard and produces no duplicate effect. If only receipt admission committed, the applied guard has not yet been consumed. Treating that state as completed would omit the first application rather than prove idempotence.
\end{proof}

\begin{corollary}[Conflict safety]
If conflicting evidence cannot enter the admitted branch, a conflict cannot reduce debt, release residual, or release collateral.
\end{corollary}

\subsection{Residual and collateral release}
Settlement-confirmed value follows the debt-first rule from Paper II \citep{nechepurenko2026credit}. Let $y\ge0$ be admitted settled cash, $D\ge0$ the outstanding debt due under the registered repayment rule (including accrued interest classified as debt), and $F\ge0$ separately admitted fees. Amounts in $D$ are not counted again in $F$. Then
\[
y^D=\min\{y,D\},\qquad y^F=\min\{[y-y^D]^+,F\},\qquad y^T=[y-y^D-y^F]^+.
\]
The service cannot infer debt clearance from a local projection. Release is allowed only if canonical getters prove zero principal, zero accrued interest, no pending residual precondition, and a final debt-free lifecycle state.

\section{Canonical loss waterfall}
A liquidation request and an actual allocation are different operations. The liquidation manager records the lifecycle and execution evidence; the canonical reserve, LBP, and Senior modules own the capital balances.

Let the recognized shortfall be $\ell\ge0$ and the ordered capacities be
\[
(c_0,c_R,c_J,c_G,c_S),
\]
where the capacities are nonnegative and non-overlapping: $c_0$ is position-specific equity or buffer, $c_R$ segregated reserve, $c_J$ eligible LBP capital, $c_G$ a configured global reserve, and $c_S$ approved Senior loss capacity. The shortfall $\ell$ is measured before these listed layers are applied; a buffer already netted from the shortfall is not counted again. In the initial local deployment, absent $c_0$ and $c_G$ are explicit zero layers.

Define recursively
\[
a_k=\min\left\{c_k,\left[\ell-\sum_{j<k}a_j\right]^+\right\}.
\]

\begin{proposition}[Unique ordered allocation]
The vector $a=(a_0,a_R,a_J,a_G,a_S)$ is unique, satisfies $0\le a_k\le c_k$, preserves priority, and obeys
\[
\sum_k a_k+\ell_{\mathrm{res}}=\ell,
\]
where $\ell_{\mathrm{res}}\ge0$ is unabsorbed loss.
\end{proposition}
\begin{proof}
The recursion uniquely determines each $a_k$ from prior allocations. The minimum enforces capacity, the positive part enforces non-negativity, and later layers receive positive loss only after earlier capacities are exhausted. Conservation follows by definition of the residual.
\end{proof}

\begin{theorem}[Atomic multi-ledger waterfall]
Suppose a storage-free controller computes $a$ through a pure, versioned allocation library and invokes canonical reserve, LBP, and Senior loss selectors in one EVM transaction. Assume that every successful canonical call applies exactly its prescribed component of $a$; every mismatch, insufficient authorization, or required-call failure propagates to a revert of the root transaction; and no later action in that transaction can cancel or distort an allocation component already applied. Then the committed post-state either reflects the entire ordered allocation $a$ or no layer mutation. The controller cannot retain loss capital.
\end{theorem}

The root-revert condition is essential: a controller that catches or ignores a failed child call does not satisfy the theorem's assumptions. Under the stated conditions, EVM transaction atomicity rolls back every earlier layer mutation whenever any required component cannot be applied; on success, the exact-component and non-distortion conditions preserve the ordered allocation computed in Proposition 7.1.

The explicit layer ordering and zero-layer visibility are illustrated in \cref{fig:waterfall}.

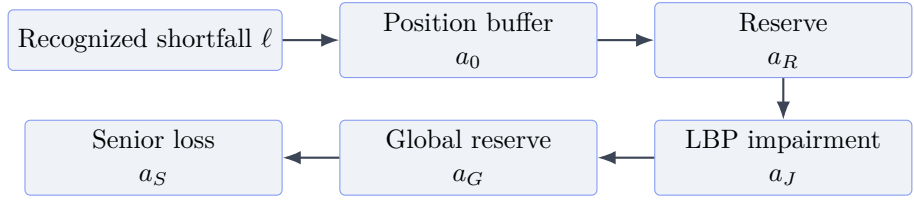
\begin{figure}[t]
\centering
\begin{tikzpicture}[node distance=0.55cm and 0.75cm,
  box/.style={draw=axblue!65,fill=axgray,rounded corners=2pt,minimum width=3.4cm,minimum height=0.8cm,align=center,font=\small},
  arr/.style={-{Latex[length=2.4mm]},thick,draw=axnavy!80}]
\node[box] (loss) {Recognized shortfall $\ell$};
\node[box,right=of loss] (buf) {Position buffer\\$a_0$};
\node[box,right=of buf] (res) {Reserve\\$a_R$};
\node[box,below=of res] (lbp) {LBP impairment\\$a_J$};
\node[box,left=of lbp] (glob) {Global reserve\\$a_G$};
\node[box,left=of glob] (sen) {Senior loss\\$a_S$};
\draw[arr] (loss)--(buf);\draw[arr](buf)--(res);\draw[arr](res)--(lbp);\draw[arr](lbp)--(glob);\draw[arr](glob)--(sen);
\end{tikzpicture}
\caption{Canonical ordered waterfall. Zero or absent layers remain explicit in the deployment manifest rather than being silently omitted.}
\label{fig:waterfall}
\end{figure}

The distinction between eligible loss-absorbing junior capital $J^{abs}$ and withdrawal-service headroom $J^{free}$ is load-bearing. The waterfall may use $J^{abs}$, not an operational liquidity proxy. A unit of LBP capital cannot simultaneously count as protection capacity, free withdrawal liquidity, and another certificate allocation.

\section{Loss-participating withdrawal queues}
A withdrawal request is denominated in shares. Until actual service and burn, queued shares remain in the canonical share supply and participate in subsequent impairment. This property prevents an already recognized loss from being transferred to slower providers by a first-exit request \citep{nechepurenko2026credit}.

Let $z_i$ be shares owned by provider $i$, $q_i\le z_i$ queued shares, and $V$ provider NAV. The queued claim at time $t$ is not a fixed cash amount; it is
\[
W_{i,t}=q_i\pi_t,\qquad \pi_t=\frac{V_t}{Z_t},
\]
for $Z_t>0$, subject to holdbacks, encumbrance, and available cash at service. A loss before service changes $V_t$ and therefore changes the queued claim. Capital still reserved to absorb loss is not automatically available for withdrawal.

\begin{proposition}[Queue participation]
If queued shares remain included in $Z_t$ until service and burn, and losses reduce canonical NAV before service, then requesting withdrawal cannot exempt the queued shares from that loss.
\end{proposition}

The accepted FIA-09 record exercises the actual Senior \texttt{CreditPool} queue, not a bridge-owned mapping. The analogous \texttt{LBPVault} queue remains a successor evidence item rather than part of the accepted record. The evidence must show ordering, queue head, share supply, loss application, service, and burn from canonical views and events.

\section{Protocol-wide pause and governed recovery}
The protocol has multiple modules with separate pause states. A local bridge-level flag is not sufficient: risk-increasing calls could remain available on an unpaused canonical module.

Let $\mathcal T$ be the manifest-frozen target set. Each mutating selector is classified as
\[
\textsf{RISK\_INCREASING},\quad \textsf{RISK\_REDUCING},\quad \textsf{RECOVERY\_ONLY},\quad \textsf{ADMIN\_ONLY}.
\]
The guardian may atomically pause risk-increasing operations across $\mathcal T$. The coordinator stores no pause flag; the authoritative state remains in each target module. Unpause requires a queued and executed \texttt{GovernanceTimelock} operation.

The execution policy also distinguishes two order treatments. Under \texttt{PERSIST\_THROUGH\allowbreak\_PAUSE}, an already-submitted order may remain only when its persistence is demonstrably risk-reducing under the paused capability state. Under \texttt{CANCEL\_ON\allowbreak\_PAUSE}, an order that could increase risk is cancelled, but such an order cannot support a certificate whose hard-flat horizon spans the pause. A pause is not a book reset: reactivation with a new book epoch and post-close lock are separate capability transitions.

\begin{proposition}[Pause-policy certificate consistency]
A hard-flat or reduce-only certificate may rely on an order across a pause only under \texttt{PERSIST\_THROUGH\_PAUSE} and only if the order remains risk-reducing in the registered pause mode. If policy requires \texttt{CANCEL\_ON\_PAUSE}, the order contributes zero certified capacity beyond the pause boundary.
\end{proposition}

\begin{theorem}[Atomic pause fan-out]
If the coordinator calls every target in one EVM transaction and reverts when any target fails, then the transaction cannot commit a strict subset of the target pauses.
\end{theorem}

\begin{proposition}[Recovery-path availability under selective pause]
If the selector policy preserves every transition needed by the registered risk-reducing path while blocking new risk, a pause does not by itself remove that path's callable selectors. Completion additionally requires admissible evidence, sufficient applicable capacity, canonical availability/finality, and eventual scheduling and retries. Selector availability alone does not prove eventual closure.
\end{proposition}

The emergency control and delayed recovery sequence are summarized in \cref{fig:pause}.

\begin{figure}[t]
\centering
\begin{tikzpicture}[node distance=0.55cm and 0.55cm,
  box/.style={draw=axblue!65,fill=axgray,rounded corners=2pt,text width=3.0cm,minimum height=0.75cm,align=center,font=\small},
  good/.style={draw=axteal!80!black,fill=axteal!7,rounded corners=2pt,text width=3.0cm,minimum height=0.75cm,align=center,font=\small},
  arr/.style={-{Latex[length=2.4mm]},thick,draw=axnavy!80}]
\node[box] (guardian) {Guardian detects boundary breach};
\node[box,right=of guardian] (fan) {Atomic pause fan-out\\over manifest target set};
\node[good,right=of fan] (safe) {Risk-increasing selectors blocked};
\node[good,below=of fan] (reduce) {Settlement, repayment, liquidation, reconciliation remain available};
\node[box,right=of reduce] (tl) {Timelock queues reviewed recovery};
\node[good,right=of tl] (unpause) {Targeted unpause after delay};
\draw[arr](guardian)--(fan);\draw[arr](fan)--(safe);\draw[arr](fan)--(reduce);\draw[arr](reduce)--(tl);\draw[arr](tl)--(unpause);
\end{tikzpicture}
\caption{Emergency pause and governed recovery. The coordinator fans out commands but does not own the canonical pause state.}
\label{fig:pause}
\end{figure}
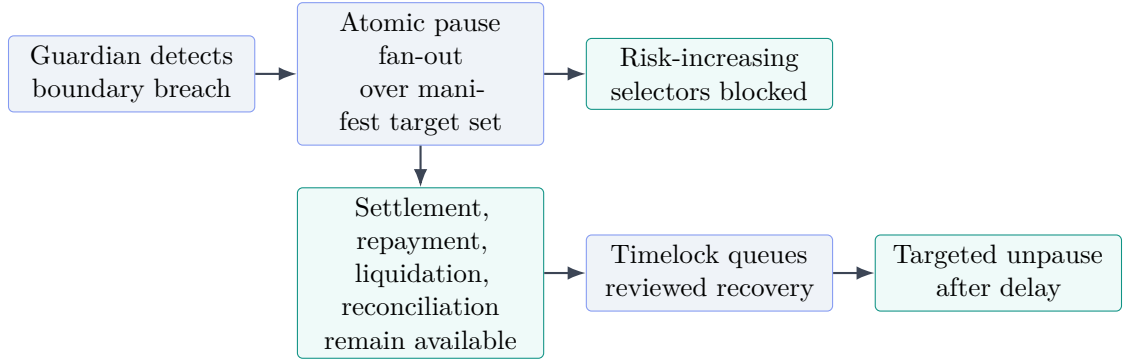

\section{Exactly-once recovery}
Distributed systems cannot guarantee exactly-once message delivery in the abstract. The relevant property is narrower: a financial transition that has already committed on-chain must not commit a second time when a worker restarts, while the off-chain projections must remain rebuildable.

Let $\tau$ be a transaction and $e$ its evidence receipt. Let $C$ denote a durable chain-confirmed checkpoint recorded before projection update. The recovery sequence is
\[
\begin{aligned}
\tau\ \text{confirmed}&\rightarrow C\ \text{persisted}\rightarrow\text{worker crash}\rightarrow\text{new worker replay}\\
&\rightarrow\text{replay/no-op}\rightarrow\text{projection rebuild}.
\end{aligned}
\]

\begin{definition}[Financial exactly-once recovery]
For a stable operation domain, key, and payload binding, recovery is financially exactly-once when there is at most one successful committed financial effect, retries are state-inert or read that effect, and derived projections are rebuildable from retained canonical observations. For an operation already known to have succeeded, the effect count is exactly one. Transaction submissions, reverted or no-op attempts, acknowledgements, and distinct settlement/release operations are not duplicate financial effects.
\end{definition}

\begin{theorem}[Checkpointed no-reexecution]
Assume the confirmed checkpoint binds the manifest, stable operation identity, payload, committed application, receipt, and evidence identifier; retries query canonical application state; and the financial selector enforces the atomic application guard. A restart after that checkpoint cannot commit a duplicate financial effect for the same operation.
\end{theorem}
\begin{proof}
The checkpoint identifies an already-committed application. Canonical lookup and the same-key guard either return that result or reject/no-op a repeated application. Such attempts can have their own transaction identifiers but cannot create another committed financial delta. Projection rebuilding does not consume another financial operation.
\end{proof}

The theorem is explicitly post-checkpoint. It does not, by itself, establish recovery in every earlier crash window. In particular, an admitted receipt is not an applied settlement, and an unavailable or unfinalized lookup is not reliable evidence of absence. The at-most-once safety argument is conditional on the canonical guard even when a service checkpoint is missing; no additional crash-window execution is reported here. Eventual application requires separate liveness assumptions. A local ledger also cannot make a non-atomic external-venue effect exactly-once merely by retrying a local command.

The retained FIA-010 fixture has the narrower observation of one canonical financial transaction, zero duplicate financial delta, old/new process identities, a durable checkpoint, and rebuilt projections. That historical observation is preserved, not elevated into a universal bound on transaction submissions. Settlement application (FIA-004) and eligible residual release (FIA-005) remain distinct valid operations.

The public verifier has a narrower scope than a runtime replay: the inspected companion checks selected fields and published-file integrity, while withheld runtime identities remain attested. A source inspection at the separately pinned September-review snapshot confirms this distinction \citep{axientfiaaudit2026}; it is not a rerun of the private protocol.

\section{Deterministic deployment manifest}
The graph is identified by more than source code. Addresses, bytecode, ABIs, storage layouts, role assignments, pause targets, selectors, adapter identity, API schema, and generated-client identity collectively determine the system under test.

\begin{definition}[Canonical deployment manifest]
A canonical deployment manifest $M$ is a canonicalized record containing at least:
\begin{enumerate}[label=(\roman*)]
\item chain and execution-mode identity;
\item workspace commit or tag;
\item every deployed canonical module and controller address;
\item creation and deployed bytecode hashes;
\item ABI and storage-layout hashes;
\item role assignments and selector policy;
\item pause target set;
\item emulator adapter identity;
\item OpenAPI, schema, and generated-client hashes;
\item a separately stored manifest digest, excluded from the manifest body that it hashes.
\end{enumerate}
\end{definition}

The manifest body excludes its own digest and downstream receipts or package digests. Canonicalization operates on that body, not a self-referential wrapper. The canonical serialization must be fixed, for example RFC 8785 JSON canonicalization \citep{rfc8785}; the retained release codec and test vectors, not mathematical tuple notation alone, determine historical bytes. No existing manifest or digest is recomputed or replaced by this clarification. The full stack fails closed if runtime addresses, bytecode, role, selector, schema, or adapter identity differs from $M$. This resembles supply-chain provenance systems in which a result is meaningful only under a sealed build and dependency identity \citep{torresarias2019intoto,slsa2023,nistssdf2022}.

\begin{proposition}[Manifest-bound command safety]
If every service command checks that its target address and selector are authorized by $M$, then a command directed to a stale deployment, unregistered controller, forbidden selector, or legacy bridge is rejected before financial execution.
\end{proposition}

\section{Canonical state projection}
A final hash is useful only if its source state is well defined. Fixture expectations are permitted only as simulator outputs; they cannot be copied into an actual layer.

For scenario $s$ and terminal block $b$, define the canonical projection
\[
\Pi(s,b,M)=\bigl(M,\{V_m(b)\}_{m\in\mathcal M_c},\mathcal E(s,b)\bigr),
\]
where $V_m(b)$ is the set of authorized ABI views from canonical module $m$, and $\mathcal E(s,b)$ is the ordered event list from manifest addresses, sorted by
\[
(\textsf{block number},\textsf{transaction index},\textsf{log index}).
\]
The state hash is
\[
h_{\mathrm{contract}}=\Hh(\mathrm{JCS}(\Pi(s,b,M))).
\]

\begin{proposition}[Fixture exclusion]
If $\Pi$ is derived from authenticated, retained canonical observations at the registered boundary, and its actual-state inputs are not copied directly or indirectly from the simulator expectation, agreement of their declared semantic projections is an observed comparison. Excluding only the expected digest is insufficient: copying the expected state and rehashing it would preserve the same defect.
\end{proposition}

The remaining layers---journal, service, API, browser, and audit reconstruction---must retain their actual inputs and root ancestry. Agreement is tested on the registered common semantic projection, not on equality of heterogeneous raw artifact digests. Derivation remains necessary even when normalized values agree. Independent runtime participation would require an additional claim-specific corroboration gate; that stronger gate was not evaluated for the registered local cohort and is \texttt{NOT\_EVALUATED}, not a failure of its accepted local profile.

\section{FIA evaluation surface}
The twelve Financial Interaction Assertions form the intended end-to-end surface. \Cref{tab:fia} states the canonical questions, not final results.

\begin{longtable}{@{}P{0.09\textwidth}P{0.27\textwidth}P{0.56\textwidth}@{}}
\caption{FIA scenarios and canonical proof obligations.}\label{tab:fia}\\
\toprule
FIA & Scenario & Canonical proof obligation \\
\midrule
\endfirsthead
\toprule
FIA & Scenario & Canonical proof obligation \\
\midrule
\endhead
01 & successful admission & risk approval, atomic position/debt origination, manifest-bound receipt, projection chain \\
02 & exact $2\times$ admission & discrete leverage policy, exact calldata, same atomic graph \\
03 & risk rejection & typed rejection and zero financial mutation \\
04 & partial then final settlement & admitted partial and final evidence, bounded debt application, lifecycle closure \\
05 & residual and collateral release & release only after FIA-04 establishes debt-clearance eligibility \\
06 & liquidation to hard-flat & authorized liquidation request/execution and hard-flat transition \\
07 & liquidation and waterfall & request/execution, reserve--LBP--Senior ordering, conservation \\
08 & evidence conflict & canonical exception/quarantine and zero debt mutation \\
09 & withdrawal queue & real pool/vault queue and loss participation until service/burn \\
10 & exactly-once recovery & chain checkpoint, process restart, one financial transaction, rebuild \\
11A & stale capability before admission & VenueRegistry/RiskManager typed rejection, causal freshness check, and zero mutation \\
11B & book-epoch change with $D_t>0$ & epoch transition, risk-increase revocation, fail-closed/risk-reducing path, and retained causal evidence \\
12 & pause and recovery & frozen targets, guardian fan-out, pause-policy semantics, timelocked recovery/unpause \\
\bottomrule
\end{longtable}

Each FIA is complete only when the role-specific UI or typed test client issues a command, the API and durable saga bind the command to canonical selectors and receipts, canonical state is observed, and the journal, service, API, browser, and audit layers retain registered derivations. The parent denominator remains twelve. FIA-011 is conjunctive: it cannot pass unless both FIA-011A and FIA-011B have retained raw and normalized evidence, causal checks, negative controls, replay, and explicit report rows.

\section{Implementation mapping and verification obligations}
\subsection{FIA-to-module touch matrix}
\Cref{tab:fia-modules} states the minimum module participation required for each scenario. A final evidence package must show actual manifest-authorized calls and events, not only a test fixture that describes them.

\begin{longtable}{@{}P{0.08\textwidth}P{0.28\textwidth}P{0.29\textwidth}P{0.27\textwidth}@{}}
\caption{Minimum canonical module coverage by FIA.}\label{tab:fia-modules}\\
\toprule
FIA & Required authorities & Financial effect & Evidence emphasis \\
\midrule
\endfirsthead
\toprule
FIA & Required authorities & Financial effect & Evidence emphasis \\
\midrule
\endhead
01 & Risk, Position, Debt, Evidence & admitted $1\times$ position and debt within the $2\times$ ceiling & atomicity and correlation \\
02 & same as FIA-01 & exact discrete $2\times$ policy & exact command and calldata identity \\
03 & Risk and capability inputs & no mutation & typed rejection and before/after equality \\
04 & Evidence, Debt, Position & partial then final admitted settlement and debt/lifecycle closure & bounded application and distinct evidence identities \\
05 & Debt, Position & one-time trader residual and collateral release after debt-clearance eligibility & ordered release preconditions \\
06 & Liquidation, Position, Evidence, Debt & authorized liquidation to hard-flat & request/execution identity and close-horizon transition \\
07 & Liquidation, Reserve, LBP, Credit & ordered canonical loss allocation & conservation and allocation hash \\
08 & Evidence, Debt & exception/quarantine, zero debt delta & conflict/replay domain \\
09 & CreditPool & Senior withdrawal request remains loss-participating until service/burn & queue ordering and NAV; LBP queue is successor evidence \\
10 & Evidence, Debt, service persistence & one transaction, rebuilt projections & checkpoint and process identities \\
11A & VenueRegistry, RiskManager & no mutation & stale capability, typed cause, and freshness provenance \\
11B & VenueRegistry, RiskManager, Position, Debt & no risk increase; fail-closed/risk-reducing transition & book epoch, positive-debt causal chain, and retained transition evidence \\
12 & all frozen pause targets, Timelock & pause and governed recovery & atomic fan-out, pause-policy mode, and delayed unpause \\
\bottomrule
\end{longtable}

\subsection{Verification layers}
For every scenario, the registered local programme retains seven evidence layers with registered derivations; claims of independent runtime participation would additionally require the proposed multi-root corroboration gate, which was not evaluated for this cohort:
\begin{enumerate}
\item simulator expectation;
\item canonical contract state from ABI views and ordered canonical events;
\item append-only journal;
\item service materialization;
\item authenticated API response;
\item browser projection after a real role control or typed browser client action;
\item Audit Explorer reconstruction.
\end{enumerate}
The expected simulator hash is not an input to layers 2--7. The contract state collector accepts only addresses in the deployment manifest. Journal ordering references block, transaction, and log coordinates rather than receipt-file order.

\begin{definition}[Canonical E2E completion]
A scenario is canonically complete only if its typed command enters through the approved role boundary, every required financial effect occurs in the canonical graph, all seven layers are retained and correlated, mandatory negative controls pass, and no forbidden authority participates.
\end{definition}

\subsection{Differential verification}
The pure allocation library, reference waterfall, exact simulator, and canonical contract transition can be compared without granting the reference model financial authority. The comparison is
\[
\mathrm{normalize}(x^{\mathrm{sim}}_{s})=\mathrm{normalize}(x^{\mathrm{contract}}_{s}),
\]
for registered scenario $s$. Disagreement may reveal a contract defect, reference-model defect, or underspecified normalization. It is reported rather than averaged away.

\subsection{Failure classification}
The final report should distinguish:
\begin{itemize}
\item \texttt{FULL\_E2E}: complete canonical graph and evidence;
\item \texttt{READ\_ONLY\_E2E}: full read/reconstruction without required financial command;
\item \texttt{CANONICAL\_REPLAY\_ONLY}: deterministic contract replay without full application path;
\item \texttt{BOUNDARY\_BLOCKED}: a declared environment, data-rights, or authority boundary prevents completion;
\item \texttt{FAILED}: the intended transition or evidence property is contradicted.
\end{itemize}
These labels describe distinct profiles, not a requirement to satisfy every conceivable production gate. The accepted local result is a pass within its registered gate set. A missing required gate is incomplete inside that scope; a broader untested claim is \texttt{NOT\_EVALUATED}. An expected typed rejection can be a local pass without a committed financial transaction. No label in an immutable source schema is rewritten by this explanatory distinction. A scenario cannot be promoted by renaming a harness or copying projection hashes.

\section{Proof obligations for the final implementation}
The final source and evidence review should discharge the following obligations.

\paragraph{O1: controller storage exclusion.}
Storage layout and source review show that new controllers contain immutable module addresses and non-authoritative metadata only.

\paragraph{O2: manifest completeness.}
Every canonical module, controller, role, selector, adapter, pause target, ABI, bytecode, and schema identity used in a scenario is present in the sealed manifest.

\paragraph{O3: authority reachability.}
Every financial mutation in the scenario occurs in the declared authority module; no legacy bridge or test evidence manager participates.

\paragraph{O4: waterfall conservation.}
Before/after vault values and events satisfy the registered allocation and conservation equation.

\paragraph{O5: queue persistence.}
Queued shares remain in supply and participate in impairment until canonical service and burn.

\paragraph{O6: replay safety.}
Duplicate and conflicting evidence create no unauthorized financial delta.

\paragraph{O7: recovery exactly-once.}
A confirmed application checkpoint survives worker termination; a fresh worker reconstructs state without a duplicate committed effect. The one-transaction count remains the narrower retained FIA-010 fixture result.

\paragraph{O8: pause coverage.}
The target set is complete, fan-out atomic, risk-reducing selectors remain available where specified, and unpause requires the timelock.

\paragraph{O9: derivation topology and corroboration.}
Each evidence layer retains its inputs, derivation, and root ancestry; a fixture expectation is never copied into an actual layer. A claim of independent runtime participation requires a separately specified multi-root corroboration predicate rather than pairwise relabeling; that proposed predicate was not evaluated for this cohort.

\paragraph{O10: reproducibility.}
Two clean complete packages and the Phase 1--17 replay meet their declared raw/normalized identity policies and survive independent review.

\paragraph{O11: capability continuity.}
FIA-011A and FIA-011B each retain the exact capability/book-epoch transition, causal checks, negative controls, replay result, and explicit report coverage; the parent FIA-011 remains non-passing unless both subcases pass.

\paragraph{O12: legal authority binding.}
Every risk-increasing command binds the canonical module to an eligible entity role, custody class, and permitted-use class; the additive \texttt{LEGAL\_NO\_GO} case proves zero mutation when displayed economic capacity is legally ineligible.

\appendix
\section{Additional proof details}
\subsection{Composition of storage-free controllers}
Let $C_1$ and $C_2$ be storage-free controllers over disjoint or overlapping canonical domains. Their sequential composition $C_2\circ C_1$ does not itself create a second ledger because neither controller stores authoritative financial state. However, correctness still requires that the composed selector order satisfy the protocol invariants. Storage freedom is necessary for authority uniqueness but not sufficient for semantic correctness.

\begin{remark}[Why an omnibus controller is discouraged]
A large controller can remain storage-free and still become operationally opaque, hard to pause, and difficult to audit. The architecture therefore prefers narrow controllers aligned with one atomicity boundary. The rule is not only ``no storage'' but also ``one reviewed reason for atomic composition.''
\end{remark}

\subsection{Saga safety versus atomic safety}
Atomicity prevents partial state within one transaction. A saga permits intermediate canonical states. Its safety therefore depends on registered intermediate states, idempotency, compensation or recovery paths, and a durable progress record. In the settlement saga, an admitted evidence receipt and a chain-confirmed debt application may exist before the browser updates. That is acceptable because financial authority is already canonical and the projection is rebuildable. It would not be acceptable for the browser or service database to release collateral before canonical debt clearance.

\subsection{Exactly-once financial effect}
The exactly-once property can be restated as a linearizability condition around the chain-confirmed transaction: the financial command linearizes at the committed canonical transaction, while later retries linearize as reads/no-ops of the same completed command. The service may process the message more than once; the financial effect occurs once.

\subsection{Explicit zero layers}
Representing an absent position buffer or global reserve as an explicit zero layer prevents two ambiguities: whether the layer was intentionally disabled or accidentally omitted, and whether a later configuration can insert the layer without changing allocation semantics. The manifest therefore records both presence and configured capacity.

\section{Registered local evidence result}
\interimnote

The evidence companion documents the registered local cohort \citep{axientfia2026}. Its retained aggregate records support the counts and derivation classes evaluated here; confidential operational identifiers remain outside the distributed evidence boundary.

The completed evidence milestones and their retained identities are summarized in \cref{tab:final-evidence}.

\begin{table}[t]
\centering\small
\caption{Registered local evidence milestones.}
\label{tab:final-evidence}
\begin{tabularx}{\textwidth}{@{}P{0.29\textwidth}P{0.19\textwidth}Y@{}}
\toprule
Milestone & Observed result & Evidence description \\
\midrule
Source-owned FIA evidence & five non-omnibus adapter groups; registered local evidence reports completed coverage & publication-safe cohort record \\
All-12 cohort & 12 records, 7 retained layers each, 84 positions & registered local cohort \\
Clean materialization & two fresh trees; one byte-identical file index & registered materialization comparison \\
Strict archive & separate verifier accepted the registered archive & registered local package \\
Ordered replay & selected replay is bound to the registered archive & registered local replay \\
\bottomrule
\end{tabularx}
\end{table}

The all-12 ledger contains one uniquely correlated record for each FIA and seven retained
\texttt{ACTUAL\_EVIDENCE} positions per record. Two clean materializations share one byte-identical canonical file index, and a separate verifier accepted the registered deterministic archive. A selected ordered replay is content-bound to that archive. These findings are reported at the registered-local level; confidential materialization, package, and replay identifiers are outside the distributed evidence boundary.

FIA-011 remains one parent assertion in the denominator of twelve and retains both mandatory subcases: stale capability before admission and book-epoch change after position opening while debt remains positive. The retained A/B envelope and focused adapter tests prevent a generic FIA-03 rejection or a one-receipt placeholder from satisfying the parent.

\subsection{Claim typing}
The observed result is the completed cohort-bound evidence chain described above. Its interpretation is that the canonical protocol graph can preserve distinct financial assertions through common manifest, materialization, packaging, and replay identities. Public-network operation, external-venue execution, real custody, terminal oracle settlement, and production governance are successor evidence programmes rather than conclusions inferred from this cohort.

\section{Security analysis}
\subsection{Threat model}
The relevant adversary may:
\begin{itemize}
\item direct a service command to a stale or unauthorized address;
\item substitute a selector with the same user-facing meaning but different state effects;
\item replay or conflict settlement evidence;
\item mutate one layer of the evidence package while preserving others;
\item restart workers between chain confirmation and projection update;
\item attempt partial pause, premature unpause, or risk increase during recovery;
\item exploit a second ledger in a controller, service, or test bridge;
\item count the same reserve or LBP capital in multiple certificates or withdrawal headroom.
\end{itemize}

The architecture addresses these risks through unique authority, manifest binding, atomic controllers, receipt idempotency, canonical queues, explicit capital eligibility, and fail-closed evidence promotion. It does not solve oracle truth, venue insolvency, stablecoin failure, key compromise, legal priority, or undiscovered contract vulnerabilities.

\subsection{Upgradeability and code identity}
The public evidence companion records the tested source and manifest identities but does not convert that snapshot into a general immutability claim. A production release must still identify immutable, proxy-based, and replaceable modules and specify the authorized manifest transition under upgrades.

\subsection{Economic non-duplication}
The current stateful integrated waterfall remains a reference oracle rather than canonical capital storage. Differential tests may compare the pure allocation result with the reference model, but the production-path controller must mutate only the real reserve, LBP, and Senior modules. This distinction is essential: mathematically equal outputs do not justify two authoritative balance sheets.

\section{Relation to prior work}
The no-secondary-ledger principle specializes transaction-processing and linearizability ideas to modular smart-contract finance \citep{herlihy1990linearizability,gray1993transaction}. Ethereum provides atomic transaction semantics but does not automatically provide a correct multi-contract architecture \citep{wood2014yellowpaper}. DeFi lending systems demonstrate on-chain pool shares, utilization, liquidation, and reserve factors \citep{compound2019,aave2022,werner2022defi}; event-margin protocols add time-bounded execution, external settlement evidence, and debt-free finality. Financial-network models motivate explicit loss ordering and isolation \citep{eisenberg2001systemic}. The present contribution is the combination of single-domain authority, storage-free composition, durable evidence sagas, manifest-bound deployment, and independently reconstructible state across a leveraged event-market graph.

\section{Reproducibility and evidence availability}
The public evidence companion contains sanitized FIA definitions, schemas, checksums, provenance records, and an offline verifier for published structure and metadata. Raw runtime payloads and operating-system paths are excluded. The verifier recomputes repository-file checksums and checks selected required fields, counts, ordering, canonical identifiers, and internal consistency; payload-derived identities remain attested because their raw roots are not public.

\section{Limitations}
The paper establishes a formal composition architecture and an observed cohort-bound evidence result. It does not estimate production failure rates, external venue compatibility, real-asset custody, human usability, regulatory compliance, or economic profitability. Those are separate empirical and operational questions.

\section{Conclusion}
A leveraged event-market protocol is not a collection of contracts; it is a single financial state machine whose authority must remain unambiguous across modules, services, and recovery. The architecture developed here assigns one storage owner to every financial domain, confines composition logic to storage-free controllers, uses atomic transitions only where partial completion is unsafe, uses durable idempotent sagas where evidence arrives asynchronously, and binds the complete system to a deterministic manifest. The resulting design makes it possible to state exact proof obligations for origination, settlement, liquidation, queues, pause, and recovery. The registered local evidence chain shows how the twelve FIA assertions, clean materialization, strict packaging, and ordered replay can be bound without turning controllers or services into a second financial ledger.

\section*{Funding}
No external funding was received for this research.

\section*{Competing Interests}
A related pending UAE patent application is disclosed in the Patent Disclosure section.

\section*{Data and Code Availability}
The article reports a registered local engineering-evidence cohort. The retained public evidence companion supports recomputation of its published structure and metadata; it is not an external dataset deposit, a production record, or a live-venue data source.

\noindent\textbf{JEL:} C63, G13, G14, G23.

\section*{AI Disclosure}
OpenAI ChatGPT and Codex were used for editorial and technical assistance during manuscript preparation. The author made all substantive research decisions, reviewed the final manuscript, and assumes full responsibility for its contents.

\section*{Patent Disclosure}
Certain technical aspects described in this article are included in a UAE patent application filed on 30 September 2026 (Application No. MOE-Patents-PTA-0003399-20260930). The application is pending; no granted patent is claimed.

\FloatBarrier
\begingroup
\raggedright
\bibliographystyle{plainnat}
\setlength{\bibsep}{4pt plus 0.3ex}
\bibliography{references}
\endgroup
\end{document}